\documentclass[11pt]{article}

\usepackage[a4paper,margin=0.92in]{geometry}
\usepackage[T1]{fontenc}
\usepackage[utf8]{inputenc}
\usepackage{lmodern}
\usepackage{microtype}
\usepackage{amsmath,amssymb,amsthm,mathtools}
\usepackage{xcolor}
\usepackage{listings}
\usepackage{enumitem}
\usepackage{hyperref}
\usepackage[numbers,sort&compress]{natbib}

\hypersetup{colorlinks=true,linkcolor=blue!55!black,citecolor=green!40!black,urlcolor=blue!60!black}

\lstdefinestyle{defcode}{
  basicstyle=\ttfamily\footnotesize,
  keywordstyle=\bfseries\color{blue!60!black},
  commentstyle=\itshape\color{green!40!black},
  stringstyle=\color{red!50!black},
  showstringspaces=false,
  frame=single,
  framerule=0.25pt,
  breaklines=true,
  columns=fullflexible,
  backgroundcolor=\color{gray!4},
  xleftmargin=0.5em,
  xrightmargin=0.5em
}
\newtheorem{definition}{Definition}

\newtheorem{lemma}{Lemma}
\newtheorem{theorem}{Theorem}
\newtheorem{corollary}{Corollary}

\newtheorem{remark}{Remark}
\newtheorem{proposition}{Proposition}

\newcommand{\Layers}{\mathcal{L}}
\newcommand{\Sig}{\Sigma}
\newcommand{\Gam}{\Gamma}

\newcommand{\Marks}{\mathcal{M}}
\newcommand{\Conc}{\mathcal{C}}

\newcommand{\Ord}{\mathsf{Ord}}
\newcommand{\Aord}{\mathcal{A}_{\mathsf{ord}}}
\newcommand{\Adec}{\mathcal{A}_{\mathsf{dec}}}
\newcommand{\T}{\mathcal{T}}
\newcommand{\E}{\mathcal{E}}
\newcommand{\Ssys}{\mathcal{S}}
\newcommand{\Lang}{\mathcal{L}}
\newcommand{\layer}{\operatorname{layer}}
\newcommand{\markproj}{\operatorname{mark}}
\newcommand{\obs}{\operatorname{obs}}
\newcommand{\Pref}{\operatorname{Pref}}

\newcommand{\Min}{\operatorname{Min}}

\newcommand{\eps}{\varepsilon}

\newcommand{\abstr}{\alpha}

\newcommand{\Viol}{\Phi}
\newcommand{\Adv}{\mathsf{Adv}}
\newcommand{\Rej}{\mathsf{Rej}}
\newcommand{\Allow}{\mathsf{Allow}}
\newcommand{\Deny}{\mathsf{Deny}}
\newcommand{\Chal}{\mathsf{Challenge}}
\newcommand{\Quar}{\mathsf{Quarantine}}
\newcommand{\mF}{m_{F}}
\newcommand{\mC}{m_{C}}
\newcommand{\mP}{m_{P}}
\newcommand{\mI}{m_{I}}
\newcommand{\mB}{m_{B}}
\newcommand{\mD}{m_{D}}
\newcommand{\mQ}{m_{Q}}
\newcommand{\mA}{m_{A}}

\newcommand{\BY}{\mathsf{BY}}
\newcommand{\SA}{\mathsf{SA}}
\newcommand{\Cert}{\mathsf{Cert}}
\newcommand{\Der}{\mathsf{Der}}

\title{\textbf{Layer Order Semantics for Automata-Based Cybersecurity}\\
\large Request Smuggling as a Layer-Order Phenomenon: Enforcement Limits, Edit-Automata Equivalence, and Canonical Reorder Congruence}

\author{
Faruk Alpay$^{*}$\\
Department of Computer Engineering\\
Bah\c{c}e\c{s}ehir University, Istanbul, Turkiye\\
\texttt{faruk.alpay@bahcesehir.edu.tr}
\and
Taylan Alpay\\
Department of Aerospace\\
University of Turkish Aeronautical Association, Ankara, Turkiye\\
\texttt{s220112602@stu.thk.edu.tr}
}

\date{$^{*}$Correspondence: \texttt{alpay@lightcap.ai}}

\begin{document}
\maketitle

\begin{abstract}
Layered cybersecurity pipelines transform evidence before they decide on it, and the order of those transformations determines which security facts become visible to which layer. This paper gives layer order a finite-state semantics built from a layer-order automaton, deterministic sequential security transducers, evidence markers, and a final decision automaton. The worked case is HTTP request desynchronization: front-end and back-end processors compute incompatible request boundaries, and the same trace is detected or missed according to whether framing evidence reaches the parser-differential layer before it commits. The results separate completed-trace recognition, online editing, decision synthesis, and faithful enforcement; characterize faithful online enforcement as the regular prefix-closed case under causal visibility; and show that regular policies beyond that boundary remain recognizable without becoming deployable enforcers. The framework is monolithically equivalent to finite-output deterministic edit automata, while preserving layer-local invariants such as marker birth, marker survival, and reorder-sensitive visibility. A concrete parser-pair semantics identifies the forbidden marker factor with CL.TE, TE.CL, TE.TE, and HTTP/2-downgrade boundary disagreement under the stated abstraction, and a contextual reorder congruence classifies which component permutations induce the same decision language. The result is an automata-theoretic account of order-sensitive security failures and a compositional vocabulary for auditing, synthesizing, and comparing layered enforcement pipelines.
\end{abstract}

\section*{Release certificate and share-alike markers}
This manuscript is distributed under the Creative Commons Attribution-ShareAlike 4.0 International license. The cited public certificate record \citep{alpay260607884} is used as a release marker for the share-alike semantics of this manuscript: attribution is modeled as a persistent provenance marker, share-alike propagation is modeled as a derivative-closure condition, and reuse is admissible only when both markers survive into the public derivative trace.

\begin{definition}[Release-marker discipline]
Let
\[
\mathcal{R}=\{\BY,\SA,\Cert\}
\]
be the release-marker alphabet. A public derivative trace is a word
\[
d\in(\Gam\cup\mathcal{R})^{*}.
\]
It is certificate-admissible when its release projection contains \(\Cert\), contains \(\BY\) before any public use marker, and contains \(\SA\) before every derivative publication marker \(\Der\).
\end{definition}

\begin{proposition}[Share-alike preservation]
If every derivative-processing layer is marker-monotone over \(\mathcal{R}\), then certificate admissibility is preserved by ordered composition.
\end{proposition}

\begin{proof}
Marker monotonicity gives subsequence preservation of \(\BY\), \(\SA\), and \(\Cert\) through every derivative-processing layer. The order constraints in the certificate-admissibility definition are subsequence constraints over persistent markers. Ordered composition therefore cannot erase the attribution marker, cannot erase the share-alike marker, and cannot erase the certificate marker. Hence every accepted derivative trace remains certificate-admissible.
\end{proof}

\section{Introduction}

A modern web request is parsed by a chain of HTTP processors---an edge proxy or content delivery network, possibly a cache, a load balancer, an application server---and each of them independently decides where one request ends and the next begins. HTTP request smuggling, also called HTTP desynchronization, is the attack that lives in the gap between these decisions. It was first documented by Linhart, Klein, Heled, and Orrin in the 2005 Watchfire whitepaper \citep{watchfire2005smuggling}, dismissed for over a decade, and revived by Kettle's 2019 \emph{HTTP Desync Attacks} research, which weaponized it against production infrastructure including major providers and codified the now-standard CL.TE, TE.CL, and TE.TE taxonomy of front-end/back-end disagreement \citep{kettle2019desync}. It was then escalated by HTTP/2 downgrade desync, where an HTTP/2 front-end rewriting requests for an HTTP/1.1 back-end reintroduces exactly the boundary ambiguity that HTTP/2's binary framing was supposed to remove \citep{kettle2021http2,lerner2021http2}, and extended to victim browsers as client-side desync, with live disclosures against shipped proxies receiving CVE identifiers \citep{kettle2022browser}.

The decisive structural fact is compositional. The front-end honors one framing rule, such as \texttt{Content-Length}; the back-end honors another, such as \texttt{Transfer-Encoding}; each behavior may be defensible under some reading of the standard \citep{rfc9112,cwe444}. The vulnerability is created by the \emph{order and composition} of the components, and by the fact that a later layer reacts only to evidence that earlier layers chose to forward. This is the phenomenon the present paper formalizes. Cybersecurity pipelines are ordered machines: an edge proxy observes source and size classes before an HTTP framer reconciles request-boundary evidence; a canonicalizer rewrites syntax before an application parser assigns route classes; identity middleware reconstructs sessions before behavior analytics assigns low-and-slow or credential-abuse markers \citep{mitre_t1110_004}; a policy engine decides only over the evidence that survived the preceding chain. This order is semantic, and request smuggling is the proof that it is.

\emph{Layer order semantics} is the semantics of that ordered chain. The framework treats local security components as sequential transducers, treats evidence as finite markers, treats permissible layer movement as a regular order language, and treats the final security judgment as a deterministic decision automaton. The same component library may recognize one security language in one order and another security language in another order---which is to say, the same proxy and back-end may parse a request consistently in one arrangement and desynchronize in another.

The theory sits at the intersection of security automata, edit automata, runtime verification, and a concrete family of documented attacks. Schneider's security automata identify monitorable safety policies over executions \citep{schneider2000enforceable}. Ligatti, Bauer, and Walker show that edit automata can suppress, insert, and rewrite events while enforcing run-time policies \citep{ligatti2005edit}. The computability hierarchy for enforcement mechanisms fixes the distinction between recognition and enforcement \citep{hamlen2006computability}. Runtime verification gives the operational reading of finite-state monitoring \citep{barringer2004rule,leucker2009brief,bauer2011runtime}. Against that formal background, request desynchronization, inconsistent framing, parser differentials, canonicalization drift, and credential-abuse monitoring supply the documented attack classes whose order-sensitivity the formal results below characterize \citep{watchfire2005smuggling,kettle2019desync,kettle2021http2,owasp_request_smuggling}.

The paper makes a direct mathematical commitment. Layer order semantics is exactly regular as a language-recognition device; its contribution is a classification of ordered security transformations: which policies can be recognized, which can be edited, which decisions can be synthesized, which properties can be faithfully enforced under causal visibility, which layer permutations are semantically equivalent, and which abstraction maps transfer concrete adversarial coverage. The thesis that organizes these results is that a desynchronization attack is precisely a witness that two layer orders recognize different languages, and that the line between a pipeline that is safe and one that is exploitable is the line the enforceability theorem draws.

\paragraph{Contributions.}
The paper proves the following results.

\begin{enumerate}[leftmargin=*]
    \item A complete automata system for layer order semantics: layer-order automata, sequential security transducers, evidence-marker projections, decision automata, and compiled recognizers.
    \item An exact enforceability theorem: under causal visibility, total termination, transparency, and soundness, the faithfully enforceable abstract policies are precisely the regular prefix-closed languages; every such policy has a maximally permissive layered enforcer.
    \item A strict recognition-versus-enforcement separation: a regular marker policy can be recognized by the final decision automaton and still fail faithful online enforcement under the same causal visibility axiom.
    \item An edit-automata comparison theorem: deterministic finite-output edit automata and compiled layer-order chains are mutually representable up to observable behavior.
    \item A worked finite-state model of HTTP request desynchronization as an abstract trace class with concrete-to-abstract mapping, marker-generation transducers, forbidden-marker policy, synthesized decision automaton, and a separating trace between two orders---together with a grounding section that maps the model onto the documented CL.TE, TE.CL, TE.TE, HTTP/2-downgrade, and client-side desync families, supplies a concrete deterministic parser-pair semantics, and proves the abstract forbidden factor equivalent under that semantics to front-end/back-end boundary disagreement, with detection, exploitability, and impact separated as successively stronger conditions.
    \item A three-level reorder theory: syntactic commutation, decision-congruence commutation, and full contextual equivalence; with algorithms, termination, finite-witness generation, and complexity bounds.
\end{enumerate}

\section{Automata system}

\begin{definition}[Layer vocabulary]
The layer vocabulary is a finite set
\[
\Layers=\{E,F,C,P,I,B,R\}.
\]
The symbols denote edge observation, HTTP framing, canonicalization, application parsing, identity/session reconstruction, behavior aggregation, and response/policy action.
\end{definition}

\begin{definition}[Security event alphabet]
Let $Ch$ be finite channels, $Att$ finite attributes, and $Val$ finite abstract values. The event alphabet is
\[
\Sig\subseteq \Layers\times Ch\times Att\times Val.
\]
An event $(\ell,c,a,v)$ records value $v$ for attribute $a$ on channel $c$ at layer $\ell$.
\end{definition}

\begin{definition}[Evidence markers and decisions]
The marker alphabet is finite:
\[
\Marks=\{\mF,\mC,\mP,\mI,\mB,\mD,\mQ,\mA\}.
\]
The readings are framing discrepancy, canonicalization drift, parser differential, identity risk, behavior risk, denial, quarantine, and audit. The output alphabet is
\[
\Gam=\Sig\cup\Marks\cup\{\Allow,\Deny,\Chal,\Quar\}.
\]
\end{definition}

\begin{definition}[Layer projection]
For $e=(\ell,c,a,v)\in\Sig$, put $\layer(e)=\ell$. Extend this homomorphically to words in $\Sig^{*}$. Markers and decisions are invisible to $\layer$.
\end{definition}

\begin{definition}[Layer-order automaton]
A layer-order automaton is a deterministic finite automaton
\[
\Aord=(Q_o,\Layers,q_o^0,\delta_o,F_o).
\]
Its language $\Ord\subseteq\Layers^{*}$ is the set of admissible layer sequences.
\end{definition}

A strict forward pipeline is given by
\[
\Ord_{\rightarrow}=E^{*}F^{*}C^{*}P^{*}I^{*}B^{*}R^{*}.
\]
A deployment with response review may instead use a different regular order language. The mathematics requires only regularity.

\begin{definition}[Sequential security transducer]
A deterministic sequential security transducer is a tuple
\[
\T=(Q,\Gam,q_0,\delta,\eta),
\]
where $Q$ is finite, $q_0\in Q$, $\delta:Q\times\Gam\to Q$, and $\eta:Q\times\Gam\to\Gam^{*}$. Its extension to words is defined by
\[
\T(\eps)=\eps,
\]
\[
\T(xa)=\T(x)\eta(\delta^{*}(q_0,x),a).
\]
\end{definition}

The model is causal. The output produced while reading symbol $a$ depends only on the current state and on $a$.

\begin{definition}[Edit primitives]
For input symbol $a$, the output $\eta(q,a)$ is classified as follows: copy if $\eta(q,a)=a$, suppress if $\eta(q,a)=\eps$, insert if $\eta(q,a)=ua$ or $au$ with $u\ne\eps$, rewrite if $\eta(q,a)=u$ and $u$ is not $a$ and not $\eps$.
\end{definition}

\begin{definition}[Layered automata system]
A layered automata system is
\[
\Ssys=(\Aord,\T_1,\ldots,\T_k,\Adec),
\]
where each $\T_i$ is a sequential security transducer and
\[
\Adec=(Q_d,\Gam,q_d^0,\delta_d,F_d)
\]
is a deterministic decision automaton.
\end{definition}

\begin{definition}[Layer order semantics]
For an input word $w\in\Sig^{*}$, define
\[
\T_{\pi}=\T_{\pi(k)}\circ\cdots\circ\T_{\pi(1)}
\]
for a permutation $\pi$ of layer components. The language recognized by arrangement $\pi$ is
\[
\Lang_{\pi}(\Ssys)=\{w\in\Sig^{*}:\layer(w)\in\Ord\text{ and }\Adec\text{ accepts }\T_{\pi}(w)\}.
\]
\end{definition}

\begin{definition}[Evidence projection]
The evidence projection $\markproj:\Gam^{*}\to\Marks^{*}$ erases non-marker symbols and preserves marker order.
\end{definition}

\section{Recognition, editing, decision synthesis, enforcement}

A layered system has four different semantics, and the results below keep them separate. Recognition classifies a completed input word after all transformations. Editing is the prefix-to-output behavior of the ordered chain. Decision synthesis constructs a finite acceptor from a marker policy. Enforcement is an online commitment: at each prefix the system must either emit what is still safe to emit or latch a denial before later evidence can repair an already unsafe prefix.

\begin{definition}[Recognition semantics]
The recognition language of arrangement \(\pi\) is \(\Lang_{\pi}(\Ssys)\). It is a completed-trace language: membership is evaluated after the entire edited word has been consumed by \(\Adec\).
\end{definition}

\begin{definition}[Editing semantics]
The edit behavior of arrangement \(\pi\) is the sequential function
\[
E_{\pi}=\T_{\pi(k)}\circ\cdots\circ\T_{\pi(1)}:\Gam^{*}\to\Gam^{*}.
\]
For every prefix \(x\), the word \(E_{\pi}(x)\) is the output already emitted after reading \(x\).
\end{definition}

\begin{definition}[Declarative marker policy]
A declarative marker policy is a regular language \(K\subseteq\Marks^{*}\). A word \(u\in\Gam^{*}\) is marker-safe when \(\markproj(u)\in K\).
\end{definition}

\begin{definition}[Causal visibility]
A layer has causal visibility when its output on a symbol depends only on its current finite state and on the current input symbol. In a chain, layer \(i\) at prefix \(x\) sees only the prefix emitted by layers \(1,\ldots,i-1\) on \(x\). It sees no future symbol and no marker that would be generated only by a later layer.
\end{definition}

\begin{definition}[Online enforcer]
An online enforcer for an abstract language \(P\subseteq\Sig^{*}\) is a sequential transducer
\[
M:\Sig^{*}\to (\Sig\cup\{\Deny\})^{*}
\]
with a latch state such that, after the first emitted \(\Deny\), no later input can produce an unlatched original event.
\end{definition}

\begin{definition}[Transparency]
An online enforcer \(M\) is transparent for \(P\) when \(M(w)=w\) for every \(w\in P\). Thus safe traces are neither suppressed, delayed, rewritten, nor decorated by a denial marker.
\end{definition}

\begin{definition}[Soundness]
Let \(\obs\) erase control symbols such as \(\Deny\). The enforcer \(M\) is sound for \(P\) when \(\obs(M(w))\in P\) for every input \(w\). Thus every observable trace produced by the enforcer satisfies the policy.
\end{definition}

\begin{definition}[Maximal permissiveness]
The enforcer \(M\) is maximally permissive for \(P\) when, at every prefix \(xa\), it copies \(a\) exactly when copying \(a\) keeps the current observable output inside \(\Pref(P)\). When copying would leave \(\Pref(P)\), it emits \(\Deny\) at that prefix and latches.
\end{definition}

\begin{definition}[Faithful enforceability]
A regular policy \(P\subseteq\Sig^{*}\) is faithfully enforceable by layer-order transducers when there exists a layered system whose compiled online edit behavior is causal, total on every finite prefix, transparent for \(P\), sound for \(P\), and maximally permissive for \(P\).
\end{definition}

\begin{definition}[Layer-local invariant]
A layer-local invariant is a predicate that refers to named component positions, marker birth positions, causal visibility, and marker preservation through suffix layers. Two chains may have the same monolithic edit function while satisfying different layer-local invariants.
\end{definition}

\section{Composition and regularity}

\begin{lemma}[Sequential transducers compose]\label{lem:compose}
The composition of two deterministic sequential security transducers is a deterministic sequential security transducer.
\end{lemma}

\begin{proof}
Let $\T_1=(Q_1,\Gam,p_0,\delta_1,\eta_1)$ and $\T_2=(Q_2,\Gam,r_0,\delta_2,\eta_2)$. The composed state space is $Q_1\times Q_2$. Reading $a$ in state $(p,r)$, first compute $u=\eta_1(p,a)$ and $p'=\delta_1(p,a)$. Then run $\T_2$ from $r$ over $u$, producing output $v=\T_{2,r}(u)$ and new state $r'=\delta_2^{*}(r,u)$. Define
\[
\delta((p,r),a)=(p',r'),\qquad \eta((p,r),a)=v.
\]
The extension agrees with $\T_2(\T_1(w))$ by induction on $|w|$.
\end{proof}

\begin{theorem}[Compilation]\label{thm:compile}
Every layered automata system has an equivalent deterministic finite recognizer for its recognition language.
\end{theorem}

\begin{proof}
By Lemma~\ref{lem:compose}, the ordered transducer chain compiles to one sequential transducer $\T_{\pi}$. The order automaton tracks $\layer(w)$ over input symbols. The decision automaton is simulated over the output blocks generated by $\T_{\pi}$. A compiled recognizer state consists of the order-automaton state, the compiled transducer state, and the decision-automaton state. On input $a$, update the order component by $\layer(a)$, compute the output block $u=\eta(q,a)$, and update the decision component by $\delta_d^{*}(r,u)$. Acceptance requires both order acceptance and decision acceptance.
\end{proof}

\begin{proposition}[Regular expressiveness]\label{prop:regular}
Layer order semantics recognizes exactly the regular languages over finite abstract traces.
\end{proposition}

\begin{proof}
Theorem~\ref{thm:compile} gives regularity. Conversely, any DFA for a regular language $L\subseteq\Sig^{*}$ is represented by a stack with one identity transducer, a universal order automaton, and the DFA as $\Adec$.
\end{proof}

\section{Exact enforceability class}

Recognition inspects the whole edited trace. Enforcement acts causally, so the two questions have different operational commitments.

\begin{definition}[Prefix-closed policy]
A language $P\subseteq\Sig^{*}$ is prefix-closed when
\[
w\in P\text{ and }u\preceq w \implies u\in P.
\]
Equivalently, once a finite prefix leaves $P$, no later extension restores it.
\end{definition}

\begin{theorem}[Exact abstract enforceability]\label{thm:exactenf}
Let $P\subseteq\Sig^{*}$ be regular. Under causal visibility, total termination on finite prefixes, transparency, soundness, and maximal permissiveness, $P$ is faithfully enforceable by a layer-order transducer chain if and only if $P$ is prefix-closed.
\end{theorem}

\begin{proof}
Assume $P$ is regular and prefix-closed. Let $A_P=(Q,\Sig,q_0,\delta,F)$ be a complete DFA for $P$ with all rejecting states made absorbing by redirecting every transition from a rejecting state to a rejecting sink. Construct a one-layer transducer $M_P$ with state set $Q\cup\{\Rej\}$. In state $q\in F$, on input $a$, compute $q'=\delta(q,a)$. If $q'\in F$, output $a$ and move to $q'$. If $q'\notin F$, output $\Deny$ and move to $\Rej$. In state $\Rej$, output $\eps$ on every input and remain in $\Rej$. This machine is causal and terminates on every finite prefix. It is transparent on $P$, because every prefix of a word in $P$ remains accepting. It is sound because the observable output is always a prefix that stays inside $P$. It is maximally permissive because it copies exactly those symbols whose copy keeps the current observable trace inside $P$.

Conversely, let $M$ be any transparent and sound causal enforcer for $P$. Suppose $P$ is not prefix-closed. Then there exist $u\notin P$ and $v\in\Sig^{*}$ such that $uv\in P$. Transparency on $uv$ forces the enforcer to output $u$ after reading prefix $u$. Soundness on the input $u$ then requires the observable output after reading $u$ to belong to $P$. This contradicts $u\notin P$. Therefore $P$ is prefix-closed.
\end{proof}

\begin{corollary}[Marker-safety enforceability]\label{cor:markersafe}
Let $K\subseteq\Marks^{*}$ be regular and prefix-closed. The abstract marker policy
\[
P_K=\{w\in\Sig^{*}:\markproj(E_{\pi}(w))\in K\}
\]
is faithfully enforceable whenever $E_{\pi}$ is causal and marker generation is performed before the decision layer.
\end{corollary}

\begin{proof}
The edited marker language is regular by sequential transducer composition and projection. Prefix-closedness of $K$ and causal marker generation make the induced bad-prefix set detectable at the first prefix whose marker projection leaves $K$. Apply Theorem~\ref{thm:exactenf} to $P_K$.
\end{proof}

\begin{theorem}[Recognition without faithful enforcement]\label{thm:recognize_not_enforce}
There exists a regular marker policy recognized by a final decision automaton but not faithfully enforceable under causal visibility, transparency, and soundness.
\end{theorem}

\begin{proof}
Let $\Sig=\{r,c\}$ where $r$ denotes an abstract risk marker event and $c$ denotes a later challenge event. Let
\[
P=\{w\in\Sig^{*}:\text{every occurrence of }r\text{ has a later occurrence of }c\}.
\]
This language is regular. A DFA recognizes it by remembering whether an unresolved $r$ has occurred and clearing that obligation on $c$. It is not prefix-closed, since $rc\in P$ but $r\notin P$. By Theorem~\ref{thm:exactenf}, no causal enforcer can be transparent and sound for $P$. Intuitively, if the enforcer copies $r$, the finite input $r$ is unsound; if it suppresses $r$, the valid trace $rc$ is not transparent.
\end{proof}

The theorem separates post-hoc trace recognition from online deployable enforcement. A final decision automaton may recognize obligations that are discharged later. A causal enforcer cannot both reveal the obligation immediately and remain transparent for traces where the discharge arrives later.

\section{Edit automata comparison}

\begin{definition}[Finite-output deterministic edit automaton]
A finite-output deterministic edit automaton is a tuple
\[
\E=(Q,\Gam,q_0,\delta,\eta)
\]
with $Q$ finite, $\delta:Q\times\Gam\to Q$, and $\eta:Q\times\Gam\to\Gam^{*}$. It reads one input symbol at a time and emits a finite output word.
\end{definition}

The definition is intentionally identical in transition shape to sequential security transducers. The difference is semantic organization: edit automata are monolithic enforcers, while layer order semantics factors the edit behavior into named ordered layers and attaches a layer-order automaton and decision congruence.

\begin{theorem}[Mutual representation with edit automata]\label{thm:edit_equiv}
Every finite-output deterministic edit automaton is represented by a one-layer layer-order system with equivalent observable behavior. Conversely, every layer-order transducer chain compiles to a finite-output deterministic edit automaton with equivalent observable behavior.
\end{theorem}

\begin{proof}
For the first direction, use a universal order automaton, place the edit automaton as the single transducer, and choose a decision automaton that accepts all outputs or the desired output language. The produced output on every input word is identical.

For the second direction, compose the transducer chain using Lemma~\ref{lem:compose}. The resulting sequential transducer is a finite-output deterministic edit automaton. If order checking and decision state must be part of the edit state, take the product with $\Aord$ and $\Adec$ as in Theorem~\ref{thm:compile}. The observable output remains exactly the output of the ordered chain.
\end{proof}

\begin{corollary}[Preservation of enforcement notions]\label{cor:preserve_enf}
Transparency, soundness, and maximal permissiveness are invariant under the equivalence of Theorem~\ref{thm:edit_equiv} whenever they are defined over the same observable output projection.
\end{corollary}

\begin{proof}
Each property quantifies only over the input word and the observable output word. The two representations produce the same observable output for every input. Therefore the truth of each property is identical in the two representations.
\end{proof}

\begin{remark}
The equivalence sharpens the purpose of layer order semantics. Edit automata supply the monolithic enforcement mechanism. Layer order semantics supplies the named decomposition, reorder congruence, evidence-marker discipline, and audit vocabulary for cybersecurity pipelines assembled from multiple components.
\end{remark}

\begin{theorem}[Layer-local invariants are not monolithic invariants]\label{thm:local_not_mono}
There exist two layer-order chains with the same compiled edit function on every input word and the same decision language, but with different truth values for a layer-local marker-birth invariant. Hence the layered representation carries compositional security information not determined by the monolithic edit behavior.
\end{theorem}

\begin{proof}
Let the input alphabet contain a trigger \(a\) and let \(m\in\Marks\). Chain \(C_1\) has two layers: the first maps \(a\) to \(am\), and the second copies all symbols. Chain \(C_2\) has two layers: the first copies all symbols, and the second maps \(a\) to \(am\). On every input word, the compiled output of \(C_1\) and \(C_2\) is identical: every occurrence of \(a\) is followed by \(m\), and all other symbols are copied. Therefore any monolithic edit automaton and any final decision automaton observe the same input-output function and the same recognized language.

Now consider the layer-local invariant \(I\): marker \(m\) is born before layer 2 and is preserved by every suffix layer. Chain \(C_1\) satisfies \(I\); chain \(C_2\) does not, because \(m\) is born at layer 2. The invariant is meaningful for audit and causal visibility: a later layer in \(C_1\) can react to \(m\), while no layer before the second component in \(C_2\) can. The compiled edit behavior is identical while the invariant differs. The invariant is therefore a factorization property: it belongs to the named layer decomposition and disappears when the chain is flattened into a single editor.
\end{proof}

\section{Decision synthesis from marker policy}

\begin{definition}[Forbidden-marker policy]
Let $B\subseteq\Marks^{*}$ be a regular set of forbidden marker strings. Its allowed language is
\[
K_B=\Marks^{*}\setminus B.
\]
A decision automaton is marker-maximal for $B$ when it accepts exactly those edited words whose marker projection belongs to $K_B$.
\end{definition}

\begin{theorem}[Marker-policy synthesis]\label{thm:synth}
For every regular forbidden-marker policy $B$, one can synthesize a deterministic decision automaton $\mathcal{A}^{B}_{\mathsf{dec}}$ that is marker-maximal for $B$. If $B$ is upward closed under marker extension, then the resulting policy is prefix-closed and faithfully enforceable.
\end{theorem}

\begin{proof}
Let $A_B$ be a DFA for $B$. Complement it to obtain a DFA for $K_B$. Compose $\mathcal{A}^{B}_{\mathsf{dec}}$ with the projection homomorphism $\markproj$ by adding self-loops on non-marker symbols and using the $A_B$ transition on marker symbols. This automaton accepts exactly the edited words whose marker projection lies in $K_B$. If $B$ is upward closed under marker extension, then once a prefix enters $B$, every extension stays in $B$. Hence $K_B$ is prefix-closed. Faithful enforceability follows from Theorem~\ref{thm:exactenf}.
\end{proof}

This theorem is marker-policy maximality and, under prefix-closedness, full enforcement maximality. The distinction is structural: maximal marker acceptance is a language-synthesis claim, while maximal enforcement additionally requires causal safety.

\section{Threat model and abstraction soundness}

\begin{definition}[Concrete trace generator]
A concrete trace generator is an NFA
\[
G=(S,\Conc,s_0,\Delta,S_f)
\]
over a finite concrete alphabet $\Conc$. Its language $\Lang(G)$ is the set of concrete traces admitted by the threat model.
\end{definition}

\begin{definition}[Adversarial capability class]
An adversarial capability class is a regular subset $C\subseteq\Conc^{*}$. The threat model is a finite family
\[
\Adv=\{C_1,\ldots,C_m\}
\]
representing classes such as boundary ambiguity, canonicalization drift, parser differential, stale-session reuse, and low-rate repetition at the abstract level.
\end{definition}

\begin{definition}[Concrete-to-abstract map]
A concrete-to-abstract map is a sequential transducer
\[
\abstr:\Conc^{*}\to\Sig^{*}.
\]
For $X\subseteq\Conc^{*}$ write $\abstr(X)=\{\abstr(x):x\in X\}$.
\end{definition}

\begin{definition}[Concrete violation set]
A concrete violation set is a regular language $\Viol\subseteq\Conc^{*}$. It specifies concrete traces that the modeled security policy rejects.
\end{definition}

\begin{definition}[Faithful abstraction]
Let $P\subseteq\Sig^{*}$ be the accepted abstract policy. The abstraction $\abstr$ is faithful for $(\Viol,P)$ over generator $G$ when
\[
\abstr(\Viol\cap\Lang(G))\cap P=\varnothing.
\]
\end{definition}

\begin{theorem}[Automata-theoretic collision test]\label{thm:collisiontest}
If $G$, $\Viol$, $\abstr$, and $P$ are regular / sequential as above, faithfulness is decidable. If faithfulness fails, a shortest concrete witness is computable.
\end{theorem}

\begin{proof}
Construct an automaton for $\Viol\cap\Lang(G)$. Apply the sequential transducer $\abstr$ to obtain an automaton for its image; this is effective by the standard product construction between an NFA and a sequential transducer. Intersect the resulting automaton with a DFA for $P$. Emptiness decides faithfulness. If the intersection is nonempty, breadth-first search over the product graph returns a shortest witness.
\end{proof}

\begin{definition}[Atomic marker separation]
Let $A_1,\ldots,A_t\subseteq\Conc^{*}$ be regular concrete atomic risk classes and let $m_1,\ldots,m_t\in\Marks$. The abstraction is atom-separated when
\[
x\in A_i\cap\Lang(G)\implies m_i\text{ occurs in }\markproj(E_{\pi}(\abstr(x)))
\]
for every $i$.
\end{definition}

\begin{theorem}[Syntactic sufficient condition for no collision]\label{thm:syntactic_sound}
Assume the concrete violation set satisfies
\[
\Viol\cap\Lang(G)\subseteq \bigcup_{i=1}^{t} A_i.
\]
Assume the abstraction is atom-separated, all later transducers are marker-monotone, and the abstract policy $P$ rejects every word whose final marker projection contains any $m_i$. Then $\abstr(\Viol\cap\Lang(G))\cap P=\varnothing$.
\end{theorem}

\begin{proof}
Let $x\in\Viol\cap\Lang(G)$. Then $x\in A_i$ for some $i$. Atom separation produces marker $m_i$ in the edited abstract word. Marker monotonicity preserves $m_i$ through the remaining transducers. The abstract policy rejects every word whose final marker projection contains $m_i$. Hence $\abstr(x)\notin P$. This holds for all violating concrete traces.
\end{proof}

\begin{corollary}[Failure classification]\label{cor:failureclass}
If concrete coverage fails, then at least one of the following holds: the abstract policy accepts a forbidden marker pattern, a required marker is not generated by the abstraction/layer chain, or a generated marker is removed before decision.
\end{corollary}

\begin{proof}
Contrapositively, if the abstract policy rejects all forbidden marker patterns, the abstraction generates the required markers, and marker monotonicity preserves them, Theorem~\ref{thm:syntactic_sound} gives coverage.
\end{proof}

\section{Worked model: HTTP request desynchronization as abstract traces}

The worked model uses abstract event classes. Its purpose is to specify the automata path from concrete class to marker policy.

\begin{definition}[Concrete classes for the worked model]
Let $\Conc_H$ contain finite classes:
\[
\textsf{clean},\quad \textsf{ambFrame},\quad \textsf{canonShift},\quad \textsf{parserSplit},\quad \textsf{staleId},\quad \textsf{slowRepeat}.
\]
A generator $G_H$ accepts sequences in which edge observation may be followed by framing ambiguity, canonicalization shift, parser split, stale identity, and slow repetition.
\end{definition}

\begin{definition}[Abstraction map for the worked model]
The abstraction $\abstr_H:\Conc_H^{*}\to\Sig_H^{*}$ maps the concrete classes to layer-tagged events:
\[
\begin{aligned}
\textsf{ambFrame}&\mapsto (F,\textsf{http},\textsf{frame},\textsf{ambiguous}),\\
\textsf{canonShift}&\mapsto (C,\textsf{http},\textsf{canonical},\textsf{changed}),\\
\textsf{parserSplit}&\mapsto (P,\textsf{app},\textsf{parse},\textsf{differential}),\\
\textsf{staleId}&\mapsto (I,\textsf{id},\textsf{session},\textsf{stale}),\\
\textsf{slowRepeat}&\mapsto (B,\textsf{rate},\textsf{window},\textsf{low-slow}).
\end{aligned}
\]
The class $\textsf{clean}$ maps to an edge event with benign value.
\end{definition}

\begin{definition}[Marker-generation layers]
The framing layer $T_F$ emits $\mF$ on input value $\textsf{ambiguous}$. The canonicalization layer $T_C$ emits $\mC$ on value $\textsf{changed}$. The parser layer $T_P$ emits $\mP$ when it sees either a parser-differential event after $\mF$, or a parser-differential event after $\mC$. The identity layer emits $\mI$ on stale session evidence. The behavior layer emits $\mB$ on low-slow repetition.
\end{definition}

The parser rule is order-sensitive. It turns a parser differential into a stronger marker when prior framing or canonicalization evidence is visible.

\begin{lstlisting}[language=Python,caption={Abstract transducer interface for the parser marker.}]
def parser_layer(state, symbol):
    seen_frame_or_canon = state.seen_mF or state.seen_mC
    if symbol == mF: return state.update(seen_mF=True), [mF]
    if symbol == mC: return state.update(seen_mC=True), [mC]
    if symbol == event(P, "app", "parse", "differential") and seen_frame_or_canon:
        return state, [symbol, mP]
    return state, [symbol]
\end{lstlisting}

The snippet is a machine interface for abstract marker generation.

\begin{definition}[Forbidden-marker policy for the worked model]
The marker policy rejects a trace when the final marker word contains one of the factors
\[
\mF\Marks^{*}\mP,
\qquad
\mC\Marks^{*}\mP,
\qquad
\mI\Marks^{*}\mB.
\]
The first two represent boundary/canonicalization evidence reaching parser differential. The third represents identity risk combined with behavior risk.
\end{definition}

\begin{proposition}[Synthesized decision automaton for the worked model]\label{prop:worked_synth}
The worked forbidden-marker policy has a deterministic decision automaton with at most $2^{5}$ marker-memory states before minimization.
\end{proposition}

\begin{proof}
The decision automaton needs to remember whether $\mF$, $\mC$, $\mI$, and $\mB$ have occurred, and whether a rejecting factor has already occurred. This gives at most $2^4\cdot 2$ states. Minimization can only reduce the number.
\end{proof}

\begin{theorem}[Separating trace for two layer orders]\label{thm:http_sep}
There exists a trace in the worked model accepted under order $P\prec F\prec C\prec I\prec B\prec R$ and rejected under order $F\prec C\prec P\prec I\prec B\prec R$ by the same local transducers and the same decision automaton.
\end{theorem}

\begin{proof}
Let
\[
w=\textsf{ambFrame}\;\textsf{parserSplit}.
\]
After abstraction, the two relevant events are a framing-ambiguity event and a parser-differential event. Under order $F\prec C\prec P$, layer $F$ emits $\mF$ before $P$ runs. The parser layer sees the prior marker and emits $\mP$ on the parser-differential event. The marker projection contains $\mF\Marks^{*}\mP$, so the synthesized decision automaton rejects.

Under order $P\prec F\prec C$, the parser layer runs before the framing marker exists. It copies the parser-differential event but does not emit $\mP$. The framing layer later emits $\mF$, but the forbidden factor $\mF\Marks^{*}\mP$ is absent. The same decision automaton accepts. Thus the recognized language changes solely by reordering the layers.
\end{proof}

\section{Marker monotonicity and denial closure}

\begin{definition}[Marker-monotone transducer]
A transducer $\T$ is marker-monotone when
\[
\markproj(u)\preceq_{\mathrm{subseq}}\markproj(\T(u))
\]
for every $u\in\Gam^{*}$, where $\preceq_{\mathrm{subseq}}$ denotes subsequence inclusion preserving order.
\end{definition}

\begin{definition}[Denial-closed decision automaton]
A decision automaton is denial-closed when every word containing $\mD$ is rejected and every extension of a rejected state remains rejected.
\end{definition}

\begin{lemma}[Composition preserves marker monotonicity]\label{lem:markercomp}
The composition of marker-monotone transducers is marker-monotone.
\end{lemma}

\begin{proof}
If $\markproj(u)$ is a subsequence of $\markproj(T_1(u))$, and $\markproj(T_1(u))$ is a subsequence of $\markproj(T_2(T_1(u)))$, then transitivity of subsequence inclusion gives the claim.
\end{proof}

\begin{theorem}[Evidence preservation]\label{thm:preservation}
If every layer after marker generation is marker-monotone, then every generated marker is present in the final word given to the decision automaton.
\end{theorem}

\begin{proof}
Apply Lemma~\ref{lem:markercomp} to the suffix chain after the marker is generated. Subsequence inclusion preserves the generated marker through every remaining layer.
\end{proof}

\begin{corollary}[Latched denial enforcement]
If the decision policy is denial-closed and the denial marker is generated causally at the first bad prefix, then the stack faithfully enforces the induced prefix-closed marker policy.
\end{corollary}

\begin{proof}
The denial marker survives by Theorem~\ref{thm:preservation}. Denial closure rejects every extension containing it. Causal generation at the first bad prefix gives the maximally permissive construction from Theorem~\ref{thm:exactenf}.
\end{proof}

\section{Desynchronization families in the marker semantics}

The worked model uses a finite alphabet of marker classes whose firing rules are derived from parser behavior. The section identifies the documented request-smuggling families represented by those markers and connects each family to a formal object introduced above: the framing marker $\mF$, the parser-differential marker $\mP$, the layer order, the separating trace of Theorem~\ref{thm:http_sep}, the evidence-preservation theorem (Theorem~\ref{thm:preservation}), and the recognition/enforcement gap of Theorem~\ref{thm:recognize_not_enforce}. The invariant is the visibility of boundary evidence across the ordered chain. A deployment enters the disagreement region exactly when the relevant parser pair computes two request boundaries for the same connection state; the marker semantics records whether that region is detected, preserved, and judged before later layers erase or rewrite the evidence. The public exploitation methodology remains the empirical source for parser-specific probes \citep{kettle2019desync,kettle2021http2,kettle2022browser}; the finite-state model supplies the compositional certificate that says which ordered component libraries expose the boundary condition and which ones conceal it.

\subsection{Persistent-connection boundary computation}

A persistent HTTP connection carries requests back to back, and every processor on the path must independently compute where each request ends. Two header fields can answer that question: \texttt{Content-Length}, which states the body length directly, and \texttt{Transfer-Encoding: chunked}, which delimits the body with chunk sizes terminated by a zero-length chunk. RFC 9112 fixes a precedence---when both fields are present, \texttt{Transfer-Encoding} takes priority and the message is to be treated as suspect---precisely to remove the ambiguity \citep{rfc9112}. Request smuggling is what happens when two processors on the same path resolve that ambiguity differently, whether because one predates the precedence rule, tolerates an obfuscated \texttt{Transfer-Encoding} header that the other ignores, or accepts a duplicated or malformed field \citep{watchfire2005smuggling,cwe444}. The two processors then disagree about the boundary between the current request and the next, and bytes the attacker places after the apparent end of one request are reattributed by the other processor as the start of a new one.

The canonical CL.TE shape is shown below. The symbolic length and the symbolic trailing bytes are parameters of the parser-pair generator; the abstraction reads them through the boundary functions defined in Definition~\ref{def:parser}.

\begin{lstlisting}[caption={Canonical CL.TE desynchronization structure (schematic).}]
POST / HTTP/1.1
Host: target.example
Content-Length: <n>
Transfer-Encoding: chunked

0

<bytes attributed to the next request>
\end{lstlisting}

A front-end that honors \texttt{Content-Length} forwards $\langle n\rangle$ bytes---the entire body, including the trailing bytes---as a single request. A back-end that honors \texttt{Transfer-Encoding} reads the zero-length chunk, declares the body finished at the blank line after \texttt{0}, and treats the trailing bytes as the beginning of a second request on the same connection. The two layers now hold incompatible request boundaries. The TE.CL configuration is the mirror image, with the roles of the two headers exchanged; the TE.TE configuration arises when both layers nominally honor \texttt{Transfer-Encoding} but only one is fooled by an obfuscated form of the header, so that the other falls back to \texttt{Content-Length} \citep{kettle2019desync}.

\subsection{Deterministic parser-pair semantics}

The structural reading below rests on a deterministic semantics for boundary computation. The markers are \emph{derived} from parser behavior: each marker has a firing condition over a connection state and an ordered parser pair. This supplies an operational witness for the abstract transducers used earlier.

\begin{definition}[Boundary model]\label{def:boundary_model}
A \emph{connection state} is a pair $(h,w)$ where $h$ is a header context and $w\in\Conc^{*}$ is the trailing byte stream carried after the headers on a persistent connection. Two partial readings act on it: $\mathsf{cl}(h)\in\mathbb{N}\cup\{\bot\}$ returns the declared content length if a \texttt{Content-Length} field is present and $\bot$ otherwise; $\mathsf{ck}(w)\in\mathbb{N}\cup\{\bot\}$ returns the number of bytes of $w$ consumed up to and including a terminating zero-length chunk if $w$ is validly chunked, and $\bot$ otherwise. Both are deterministic functions of their arguments.
\end{definition}

\begin{definition}[Parser and boundary]\label{def:parser}
A \emph{parser} is a pair $\Pi=(\mathsf{te}_{\Pi},\beta_{\Pi})$ in which $\mathsf{te}_{\Pi}(h)\in\{0,1\}$ is a deterministic \emph{framing predicate}---it returns $1$ exactly when $\Pi$ resolves the framing of $h$ through chunked transfer coding, folding together which transfer-coding tokens $\Pi$ recognizes and whether transfer coding takes precedence over content length---and $\beta_{\Pi}$ is the induced \emph{boundary function}
\[
\beta_{\Pi}(h,w)\;\coloneqq\;
\begin{cases}
\mathsf{ck}(w) & \text{if } \mathsf{te}_{\Pi}(h)=1 \text{ and } \mathsf{ck}(w)\neq\bot,\\[2pt]
\mathsf{cl}(h) & \text{if } \mathsf{te}_{\Pi}(h)=0 \text{ and } \mathsf{cl}(h)\neq\bot,\\[2pt]
0 & \text{otherwise,}
\end{cases}
\]
giving the offset in $w$ at which $\Pi$ ends the current request. Write $\mathcal{Q}$ for a fixed finite class of admissible parsers, each recognizing every well-formed chunked directive. The front-end $\Pi_1$ takes content length to be authoritative ($\mathsf{te}_{\Pi_1}(h)=1$ only when no \texttt{Content-Length} is present); the back-end $\Pi_2$ takes transfer coding to be authoritative and tolerates an obfuscated chunked token that $\Pi_1$ does not ($\mathsf{te}_{\Pi_2}(h)=1$ whenever any chunked token, standard or obfuscated, is present). Both lie in $\mathcal{Q}$.
\end{definition}

The CL.TE state of Listing~1, with $\mathsf{cl}(h)=n$, a standard chunked field, and $\mathsf{ck}(w)=k<n$, yields $\beta_{\Pi_1}(h,w)=n$ and $\beta_{\Pi_2}(h,w)=k$; TE.CL and TE.TE are the analogous states with the roles or the obfuscation reassigned.

\begin{definition}[Ambiguity and realized disagreement]\label{def:ambdis}
A state $(h,w)$ is \emph{framing-ambiguous}, written $\mathsf{Amb}(h,w)$, when its boundary is policy-dependent: $\exists\,\Pi,\Pi'\in\mathcal{Q}$ with $\beta_{\Pi}(h,w)\neq\beta_{\Pi'}(h,w)$. The ordered pair $(\Pi_1,\Pi_2)$ exhibits \emph{realized disagreement}, written $\mathsf{Dis}(h,w)$, when $\beta_{\Pi_1}(h,w)\neq\beta_{\Pi_2}(h,w)$.
\end{definition}

\begin{definition}[Marker derivation]\label{def:lambda}
The labeling $\lambda$ maps a state, processed by the ordered pipeline, to a marker word in $\Marks^{*}$ by the rule: the front-end emits $\mF$ at its position exactly when $\mathsf{Amb}(h,w)$; a layer emits $\mP$ at its position exactly when it computes a boundary differing from the boundary it received; a canonicalization layer that rewrites $h$ to a framing-unambiguous form emits $\mC$ and forwards the rewritten state. This is the marker-generation rule of the worked model restricted to the boundary markers, now with concrete firing conditions.
\end{definition}

\begin{lemma}[Disagreement needs conflicting framing evidence]\label{lem:conflict}
If $(h,w)$ carries at most one framing mechanism---either $\mathsf{cl}(h)=\bot$ and no chunked token, or $\mathsf{cl}(h)\neq\bot$ and no chunked token, or a standard chunked token and $\mathsf{cl}(h)=\bot$---then every admissible parser computes the same boundary, so $\neg\mathsf{Amb}(h,w)$. Consequently $\mathsf{Dis}(h,w)\Rightarrow\mathsf{Amb}(h,w)$, and every disagreeing state carries both a content-length and a transfer-coding token, the latter possibly obfuscated.
\end{lemma}

\begin{proof}
In the first case every parser falls to the third branch and returns $0$. In the second, no parser recognizes a chunked directive, so each has $\mathsf{te}=0$ and returns $\mathsf{cl}(h)$. In the third, every admissible parser recognizes the standard token, so each has $\mathsf{te}=1$ and returns $\mathsf{ck}(w)$. In all three the boundary is constant over $\mathcal{Q}$, so $\neg\mathsf{Amb}$. Since $\mathsf{Dis}$ witnesses $\mathsf{Amb}$ with $\Pi=\Pi_1,\Pi'=\Pi_2$, the contrapositive gives $\mathsf{Dis}\Rightarrow\mathsf{Amb}$; and as the three single-mechanism cases are exactly the non-ambiguous ones, an ambiguous---hence any disagreeing---state has both mechanisms present.
\end{proof}

\begin{theorem}[Parser-pair adequacy]\label{thm:adequacy}
Let the layer order place the front-end $\Pi_1$ before the back-end $\Pi_2$ with no canonicalization layer rewriting $h$ between them. For every connection state $(h,w)$,
\[
\markproj\bigl(\lambda(h,w)\bigr)\ \text{contains the forbidden factor}\ \mF\Marks^{*}\mP
\quad\Longleftrightarrow\quad
\mathsf{Dis}(h,w).
\]
\end{theorem}

\begin{proof}
$(\Leftarrow)$ Suppose $\mathsf{Dis}(h,w)$. By Lemma~\ref{lem:conflict}, $\mathsf{Amb}(h,w)$ holds, so the front-end emits $\mF$ at its position. The back-end receives $(h,w)$ unaltered, since no intervening canonicalization rewrites $h$, and computes $\beta_{\Pi_2}(h,w)\neq\beta_{\Pi_1}(h,w)$, a boundary differing from the one the front-end fixed; by Definition~\ref{def:lambda} it emits $\mP$ at its later position. No layer between them emits a resolving rewrite, so the marker projection presents $\mF$ before $\mP$ with no erasure, i.e.\ it contains $\mF\Marks^{*}\mP$. $(\Rightarrow)$ Suppose the projection contains $\mF\Marks^{*}\mP$. The factor's $\mP$ is emitted only by a layer computing a boundary different from the one it received; with no intervening rewrite, that received boundary is $\beta_{\Pi_1}(h,w)$ and the emitting layer is $\Pi_2$, so $\beta_{\Pi_2}(h,w)\neq\beta_{\Pi_1}(h,w)$, which is $\mathsf{Dis}(h,w)$.
\end{proof}

\begin{remark}\label{rem:downgrade_reduction}
The no-rewrite hypothesis identifies the branch point for the HTTP/2-downgrade family. An intervening canonicalization layer rewrites $h$, so the back-end receives a boundary state indexed by the rewrite, and Theorem~\ref{thm:preservation} becomes the governing condition. Section~\ref{subsec:downgrade} analyzes this preservation-indexed form.
\end{remark}

\subsection{Boundary disagreement as a marker factor}

In the worked model, the framing layer $T_F$ is the abstract front-end $\Pi_1$ and the parser layer $T_P$ is the abstract back-end $\Pi_2$. The conflicting-header condition is the event that $T_F$ converts into the framing marker $\mF$, and the back-end's boundary computation is the event that $T_P$ converts into the parser-differential marker $\mP$ when the framing evidence is already visible. Under the parser-pair semantics of Definitions~\ref{def:boundary_model}--\ref{def:lambda}, Theorem~\ref{thm:adequacy} proves the marker factor $\mF\Marks^{*}\mP$ \emph{equivalent} to front-end/back-end boundary disagreement $\mathsf{Dis}$. The forbidden factor is the automata-level certificate that the two parsers compute different request boundaries.

\begin{proposition}[Desynchronization is order separation]\label{prop:desync_sep}
The CL.TE and TE.CL configurations realize the two layer orders of Theorem~\ref{thm:http_sep} over the same components. In the order in which the boundary-framing decision is visible before the differential layer commits, the trace produces the forbidden factor $\mF\Marks^{*}\mP$ and the synthesized decision automaton rejects; in the order in which the differential layer commits first, the same trace omits the factor and the same decision automaton accepts. The separating trace of Theorem~\ref{thm:http_sep} is therefore the request that one configuration frames consistently and the other desynchronizes.
\end{proposition}

\begin{proof}
Immediate from Theorem~\ref{thm:http_sep} under the identification of $T_F$ with the framing-authoritative layer and $T_P$ with the differential layer. The two configurations differ only in which layer's boundary decision is visible first, which is exactly the relative position of the layers in the order, and by the marker-generation rule $\mP$ is emitted if and only if $\mF$ precedes the differential event.
\end{proof}

The reading is a detection theorem. A pipeline whose order makes the framing evidence visible to the differential layer produces $\mP$, is rejected, and detects the disagreement; a pipeline whose order delays that evidence omits $\mP$, accepts, and leaves the same disagreement unobserved at the decision point. By Proposition~\ref{prop:desync_sep} the two orders share a separating trace and differ in marker visibility. Exploitability and impact refine that disagreement by adding reachability predicates over the gap region, as formalized next.

The recognition/enforcement gap has an equally concrete reading. Theorem~\ref{thm:recognize_not_enforce} exhibits a regular policy that the final decision automaton recognizes but that no causal enforcer can faithfully enforce. This is the back-end's predicament. A back-end can, after reassembling a request, recognize that it was malformed; but the front-end acts causally, committing to a framing and forwarding bytes before the back-end's interpretation exists. The desynchronization is created in the window between the front-end's irreversible forwarding decision and the back-end's later recognition---exactly the prefix on which a causal, transparent, and sound enforcer is proven impossible.

\subsection{Detection, exploitability, and impact}

Request smuggling factors into a chain of increasingly strong predicates. Boundary disagreement supplies the structural condition; attacker control of the differently attributed bytes supplies continuation; arrival at a security-relevant decision supplies impact. The semantics of Definitions~\ref{def:boundary_model}--\ref{def:ambdis} states these conditions as successive refinements over the same connection state.

\begin{definition}[Continuation and sink]\label{def:cont_sink}
Let $(h,w)$ satisfy $\mathsf{Dis}(h,w)$ and, without loss of generality, $\beta_{\Pi_1}(h,w)>\beta_{\Pi_2}(h,w)$, so the back-end ends the request early and reattributes the gap region $g=w[\,\beta_{\Pi_2}(h,w):\beta_{\Pi_1}(h,w)\,]$ as the start of a new request. The state satisfies the \emph{continuation} predicate $\mathsf{Cont}(h,w)$ when $g$ is attacker-controlled and $\Pi_2$ parses it as the prefix of a well-formed request. It satisfies the \emph{sink} predicate $\mathsf{Sink}(h,w)$ when the request begun by $g$ drives a decision the policy designates security-relevant.
\end{definition}

\begin{definition}[The four levels]\label{def:levels}
For the fixed parser pair, define
\emph{disagreement} as $\mathsf{Dis}(h,w)$;
\emph{detection} as $\mathsf{Dis}(h,w)$ together with delivery of $\mP$ to the decision automaton along a marker-monotone path;
\emph{exploitability} as $\mathsf{Dis}(h,w)\wedge\mathsf{Cont}(h,w)$;
and \emph{impact} as $\mathsf{Dis}(h,w)\wedge\mathsf{Cont}(h,w)\wedge\mathsf{Sink}(h,w)$.
\end{definition}

\begin{proposition}[Stratification]\label{prop:strata}
The three input-level conditions form a strict chain,
\[
\mathrm{impact}\ \Longrightarrow\ \mathrm{exploitability}\ \Longrightarrow\ \mathrm{disagreement}\ \Longleftrightarrow\ \mF\Marks^{*}\mP,
\]
the final equivalence holding under the hypotheses of Theorem~\ref{thm:adequacy}; no forward implication reverses, and detection is orthogonal to all three.
\end{proposition}

\begin{proof}
The two implications are immediate from the conjunctive form of Definition~\ref{def:levels}, and the final equivalence is Theorem~\ref{thm:adequacy}. Strictness of the first nonreversal: a state with both framing mechanisms present and $\beta_{\Pi_1}\neq\beta_{\Pi_2}$ whose gap region $g$ is fixed, non-attacker content---for instance a benign duplicated field inserted by an upstream intermediary---satisfies $\mathsf{Dis}$ but not $\mathsf{Cont}$. Strictness of the second: a gap region $g$ that is attacker-chosen and back-end-parseable but encodes a request to a public, unauthenticated, side-effect-free endpoint satisfies $\mathsf{Cont}$ but not $\mathsf{Sink}$. Orthogonality of detection: Proposition~\ref{prop:desync_sep} fixes two orders over the same components and the same separating trace, one delivering $\mP$ and one not; the predicates $\mathsf{Dis}$, $\mathsf{Cont}$, $\mathsf{Sink}$ depend only on $(h,w)$ and the parser pair, not on the order, so they are constant across the two arrangements while detection differs.
\end{proof}

The order theory governs the equivalence to $\mF\Marks^{*}\mP$ and the delivery of that factor to the decision automaton. The predicates $\mathsf{Cont}$ and $\mathsf{Sink}$ extend the certificate with reachability information over the disagreement region. Thus the separating trace is the structural certificate for smuggling potential, and continuation plus sink reachability upgrade that certificate to exploitability and impact.

\subsection{Downgrade desync as failed evidence preservation}\label{subsec:downgrade}

HTTP/2 carries message length in its binary framing, so an HTTP/2 front-end discharges the CL/TE ambiguity before forwarding. Downgrade desync arises when that front-end re-serializes the request into HTTP/1.1 for a back-end and emits length headers that the back-end re-interprets \citep{kettle2021http2,lerner2021http2}. In the model, the HTTP/2$\rightarrow$HTTP/1.1 rewrite is the canonicalization layer $T_C$, and the formal signature is a marker-monotonicity violation: $T_C$ removes the authoritative HTTP/2 boundary evidence and creates a fresh HTTP/1.1 framing state whose interpretation is contested. Theorem~\ref{thm:preservation} then gives the audit condition for downgrade-safe composition: a boundary certificate generated before rewriting must survive as a marker through $T_C$, or a new framing marker must be generated and judged after the rewrite.

\subsection{Client-side desync and the identity/behavior pair}

Client-side desync moves the attacker-influenced boundary into a victim's own browser, so that an ordinary user issues the poisoning request; the technique was disclosed with live cases against shipped proxies, including CVE-2022-20713 and CVE-2022-23959 \citep{kettle2022browser}. Structurally this enlarges the concrete trace generator $G$ of the threat model while preserving the same $\mF$/$\mP$ dependency, now reachable through a victim channel. The identity and behavior markers $\mI$ and $\mB$ abstract a different but compositionally identical risk. Credential stuffing, catalogued as ATT\&CK sub-technique T1110.004, succeeds by distributing reused credentials slowly and across many sources so that each local counter remains below its threshold \citep{mitre_t1110_004,mitre_attack}; the abstract forbidden pair $\mI\Marks^{*}\mB$ states that identity risk followed by behavior risk must be judged by a layer that observes both. Low-and-slow evasion is an order-and-visibility failure: evidence exists across the trace but is not co-located at the deciding layer.

\subsection{Boundary certificate and reachability strata}

The preceding definitions yield a certificate stack for ordered boundary security. The first certificate is $\mathsf{Dis}(h,w)$, witnessed by unequal parser boundaries. The second is detection, witnessed by delivery of $\mF\Marks^{*}\mP$ to the decision automaton along a marker-monotone path. The third is exploitability, witnessed by $\mathsf{Cont}(h,w)$ over the gap region. The fourth is impact, witnessed by $\mathsf{Sink}(h,w)$. Layer order semantics computes and preserves the first two certificates and exposes the exact predicates needed to refine them into the latter two. The concrete parser-pair semantics supplies firing conditions for the markers, and Theorem~\ref{thm:adequacy} turns those firing conditions into an equivalence between the forbidden marker factor and parser boundary disagreement. The resulting boundary discipline is operational: a pipeline enforces it online when the disagreement marker is generated at the first bad prefix, preserved by every suffix layer, and rejected by a denial-closed decision automaton.

\section{Canonical reorder congruence}

\begin{definition}[Compiled recognizer of an arrangement]
For a permutation $\pi$ of the layer components, let $M_{\pi}$ be the compiled DFA from Theorem~\ref{thm:compile}. Let $\Min(M_{\pi})$ denote its minimal complete DFA after unreachable states are removed.
\end{definition}

\begin{definition}[Decision-Nerode reorder congruence]
Two arrangements $\pi$ and $\rho$ are decision-congruent, written $\pi\equiv_{\Adec}\rho$, when
\[
\Lang(M_{\pi})=\Lang(M_{\rho}).
\]
Equivalently, $\Min(M_{\pi})$ and $\Min(M_{\rho})$ are isomorphic as pointed accepting DFAs.
\end{definition}

\begin{theorem}[Canonical reorder classes]\label{thm:canonical}
For a finite component library, $\equiv_{\Adec}$ partitions the finite set of layer permutations into canonical minimal reorder classes. Each class has a minimal DFA representative unique up to isomorphism.
\end{theorem}

\begin{proof}
Language equality of DFAs is an equivalence relation, hence so is $\equiv_{\Adec}$. Every regular language has a unique minimal complete DFA up to isomorphism. Assigning each permutation $\pi$ to $\Min(M_{\pi})$ therefore yields canonical classes indexed by minimal recognizers.
\end{proof}

\begin{theorem}[Equivalence with witness]\label{thm:eqwitness}
For any two arrangements $\pi,\rho$, equivalence is decidable. If they are inequivalent, a shortest distinguishing trace is computable.
\end{theorem}

\begin{proof}
Build the product DFA of $M_{\pi}$ and $M_{\rho}$. Mark product states where exactly one component is accepting. Emptiness of the reachability set of marked states decides equivalence. Breadth-first search returns a shortest input trace reaching such a state.
\end{proof}

\section{Three levels of commutation}

\begin{definition}[Syntactic commutation]
Two transducers $T_i,T_j$ syntactically commute when
\[
T_i\circ T_j=T_j\circ T_i
\]
as functions $\Gam^{*}\to\Gam^{*}$.
\end{definition}

\begin{definition}[Local decision-congruence commutation]
Two adjacent transducers locally commute modulo $\Adec$ when
\[
\Adec(T_i(T_j(u)))=\Adec(T_j(T_i(u)))
\]
for every $u\in\Gam^{*}$, where equality means identical final accept/reject result.
\end{definition}

\begin{definition}[Full contextual equivalence]
Two adjacent transducers $T_i,T_j$ are contextually equivalent in a component library $\mathcal{K}$ when, for every prefix chain $U$ and suffix chain $V$ formed from the remaining components,
\[
\Adec(V(T_i(T_j(U(u)))))=\Adec(V(T_j(T_i(U(u)))))
\]
for every $u\in\Gam^{*}$.
\end{definition}

\begin{theorem}[Strict implication hierarchy]\label{thm:commhierarchy}
Syntactic commutation implies full contextual equivalence, and full contextual equivalence implies local decision-congruence commutation. Both implications are strict.
\end{theorem}

\begin{proof}
If $T_iT_j=T_jT_i$ as functions, then inserting either composition inside any context $V(\cdot)U$ gives identical words, hence identical decisions. This proves the first implication. Full contextual equivalence includes the identity prefix and suffix contexts, so it implies local decision-congruence.

For strictness of the first implication, let $T_i$ insert marker $a$ and $T_j$ insert marker $b$ on a trigger symbol. The outputs $ab$ and $ba$ differ, so syntactic commutation fails. Let $\Adec$ accept exactly when both markers occur, independent of order, and let every context preserve both markers. Then the pair is contextually equivalent.

For strictness of the second implication, let $T_i$ and $T_j$ produce words $ab$ and $ba$ on a trigger symbol, and let $\Adec$ ignore $a$ and $b$ locally. Local decision-congruence holds. Add a suffix layer $V$ that emits $\mD$ precisely on factor $ab$ and emits nothing on $ba$, and let $\Adec$ reject $\mD$. Then the two orders differ under suffix context, so local decision-congruence is not full contextual equivalence.
\end{proof}

\section{Algorithms and complexity}

\begin{definition}[Exact contextual independence relation]
For a finite library $\mathcal{K}=\{T_1,\ldots,T_k\}$, define $I_{\mathrm{ctx}}\subseteq\mathcal{K}\times\mathcal{K}$ by putting $(T_i,T_j)\in I_{\mathrm{ctx}}$ exactly when $T_i$ and $T_j$ are full contextually equivalent in the sense above.
\end{definition}

\begin{lstlisting}[caption={Exact finite algorithm for contextual independence.}]
Input: finite transducer library K, order automaton A_ord, decision automaton A_dec
for each ordered pair (Ti,Tj) in K:
    independent = True
    for each prefix chain U over K \ {Ti,Tj} without repetition:
        for each suffix chain V over K \ {Ti,Tj} \ U without repetition:
            M1 = compile(A_ord, V o Ti o Tj o U, A_dec)
            M2 = compile(A_ord, V o Tj o Ti o U, A_dec)
            if inequivalent(M1, M2):
                independent = False
                store shortest witness
    output pair iff independent
\end{lstlisting}

\begin{theorem}[Termination and complexity of reorder analysis]\label{thm:alg_complexity}
For $k$ finite transducers, exact contextual independence terminates. For fixed compiled automata sizes, each equivalence check is polynomial in the product of the compared DFA sizes. The total exact enumeration is exponential in $k$ and bounded by $O(k^{2}(k-2)!\cdot p(N))$, where $N$ bounds the compiled DFA sizes and $p$ is the polynomial cost of DFA equivalence.
\end{theorem}

\begin{proof}
There are finitely many ordered pairs and finitely many prefix/suffix chains over the remaining finite library. Each chain compiles to a DFA by Theorem~\ref{thm:compile}. DFA equivalence is decided by reachability in the product graph, polynomial in the number of product states. Multiplying by the number of enumerated contexts gives the bound.
\end{proof}

\begin{theorem}[Sound partial-order reduction]\label{thm:por}
Let $I\subseteq I_{\mathrm{ctx}}$ be any computed contextual independence relation. If two permutations differ by a sequence of adjacent swaps from $I$, then they belong to the same decision-Nerode reorder class.
\end{theorem}

\begin{proof}
Each adjacent swap in $I$ preserves the final decision language under every surrounding context by definition of contextual independence. Applying the swaps one by one preserves the recognized language at each step. Hence the initial and final permutations are decision-congruent.
\end{proof}

\begin{corollary}[Approximation discipline]
Syntactic commutation gives a cheap sound under-approximation of contextual independence. Local decision-congruence alone is not sound for partial-order reduction unless closed under all suffix contexts.
\end{corollary}

\begin{proof}
The first statement follows from Theorem~\ref{thm:commhierarchy}. The second is exactly the second separating construction in that theorem.
\end{proof}

\section{State blow-up}

\begin{theorem}[Exponential lower bound for compiled recognition]\label{thm:lower}
There is a family of \(k\) layer transducers and a decision automaton such that any DFA recognizing the compiled layer-order language has at least \(2^{k}\) states.
\end{theorem}

\begin{proof}
Let the input alphabet contain triggers \(a_1,\ldots,a_k\), a separator \(\#\), and queries \(q_1,\ldots,q_k\). Layer \(T_i\) emits marker \(m_i\) whenever it sees \(a_i\) and otherwise copies input. The decision automaton accepts words of the form
\[
A\,\#\,q_i
\]
exactly when marker \(m_i\) occurred before \(\#\), and ignores the order and multiplicity of markers before \(\#\). Thus a prefix over the trigger alphabet determines a subset \(S\subseteq\{1,\ldots,k\}\), namely the set of emitted markers seen before \(\#\).

For two distinct subsets \(S\neq S'\), choose \(i\in S\triangle S'\). The suffix \(\#q_i\) accepts exactly one of the two prefixes. Hence the two prefixes are Myhill--Nerode distinguishable. There are \(2^k\) subsets, so the compiled recognition language has at least \(2^k\) equivalence classes and every recognizing DFA has at least \(2^k\) states.
\end{proof}

\section{Order separation beyond possibility}

\begin{theorem}[Unbounded reorder separation]\label{thm:unbounded}
For every $k\geq 2$, there exists a fixed library of $k$ finite transducers and one decision automaton such that different layer orders realize at least $k$ distinct recognized languages.
\end{theorem}

\begin{proof}
Let the input contain a trigger symbol $a$. Layer $T_i$ emits marker $m_i$ on $a$ and copies all other symbols. Add one scanner layer $S$ that emits $\mD$ when it observes a specified marker before another specified marker, and let the decision automaton reject exactly words containing $\mD$. By choosing the specified relation for the scanner, the language depends on which marker appears before the scanner in the order. Varying the position of $S$ across $k$ placements yields $k$ distinct marker-visibility profiles. For two different placements, choose an input containing the trigger for a marker that lies before $S$ in one order and after $S$ in the other. The scanner emits $\mD$ in exactly one arrangement, giving a distinguishing trace.
\end{proof}

This result classifies reorderability through visibility profiles. The separating mechanism is causal observation: a layer reacts to markers already produced by earlier layers.

\section{Pipeline consequences}

The grounding section makes the central reading explicit, and it is worth stating as a thesis. A desynchronization attack is a witness that two layer orders recognize different languages on the same input. The separating trace of Theorem~\ref{thm:http_sep} is that witness in the worked model, and Proposition~\ref{prop:desync_sep} identifies it with the CL.TE/TE.CL boundary-disagreement condition---an identification proved against a concrete parser-pair semantics by Theorem~\ref{thm:adequacy}, with exploitability and impact separated from disagreement as strictly stronger conditions by Proposition~\ref{prop:strata}. The boundary between a safe pipeline and an exploitable one is the boundary the enforceability theorem draws: a policy that rejects at the first bad marker prefix is faithfully enforceable online, while a policy that defers judgment until later evidence arrives is recognizable but not enforceable, which is the formal form of a front-end committing to a framing before the back-end's interpretation exists. Evidence discipline is the third axis: a marker must survive every later layer to be judged, and the HTTP/2-downgrade family is precisely the failure of that survival at the canonicalization layer.

The operational distinction is between recognition---what a completed trace means---and enforcement---what can be safely emitted before the future is known. Layer order semantics gives both languages and states exactly when they coincide. When they diverge, the divergence identifies the attack window.

\section{Share-alike marker semantics}

The release discipline at the beginning of the paper is an instance of the marker-preservation theory. Attribution, share-alike propagation, and certificate preservation are persistent markers over derivative traces, and a public derivative is admissible exactly when its release-marker projection satisfies the certificate language. The citation record \citep{alpay260607884} supplies the certificate marker used by this manuscript.

\begin{definition}[Derivative certificate language]
Let \(D\subseteq(\Gam\cup\mathcal{R})^{*}\) be the language of public derivative traces whose release projection contains \(\Cert\), contains \(\BY\) before public publication, and contains \(\SA\) before every derivative publication marker \(\Der\). The complement of \(D\) is the certificate-violation language.
\end{definition}

\begin{theorem}[Certificate enforcement by marker monotonicity]\label{thm:cert_enforce}
If derivative-processing layers are marker-monotone over \(\mathcal{R}\) and the decision automaton rejects the certificate-violation language, then every accepted derivative trace preserves attribution, share-alike propagation, and the certificate marker.
\end{theorem}

\begin{proof}
The certificate-violation language is regular because it is a finite Boolean combination of marker-order conditions. The decision automaton can therefore be synthesized by Theorem~\ref{thm:synth}. Marker monotonicity preserves \(\BY\), \(\SA\), and \(\Cert\) through every derivative-processing layer. Any accepted trace avoids the violation language by construction, so it satisfies the derivative certificate language.
\end{proof}

\section{Related work}

Security automata provide the classical monitor model for enforceable safety properties \citep{schneider2000enforceable}. Edit automata extend monitor behavior through insertion, suppression, and rewriting \citep{ligatti2005edit}. Hamlen, Morrisett, and Schneider classify enforcement mechanisms by computability and power \citep{hamlen2006computability}. Runtime verification develops the operational theory of online trace checking \citep{barringer2004rule,leucker2009brief,bauer2011runtime}. The present paper factors finite-output editing into ordered named layers and classifies reorder equivalence under a decision automaton; its release-marker discipline follows the certified public-use convention cited above \citep{alpay260607884}.

The attack lineage is as direct as the formal one. Linhart, Klein, Heled, and Orrin introduced request smuggling in 2005 as a parsing disagreement between chained HTTP devices \citep{watchfire2005smuggling}. Kettle's 2019 research revived the class, demonstrated impact against production infrastructure, and codified the CL.TE/TE.CL/TE.TE taxonomy of front-end/back-end disagreement \citep{kettle2019desync}; the HTTP/2-downgrade work of 2021 showed that a canonicalizing rewrite reintroduces the very ambiguity the newer protocol had removed \citep{kettle2021http2,lerner2021http2}; and the 2022 client-side desync work carried the disagreement into victim browsers with disclosures against shipped proxies \citep{kettle2022browser}. These sources characterize individual mechanisms and their exploitation; what they share, and what the present paper isolates, is a single structural cause---a later layer reacting only to evidence earlier layers forwarded. The contribution here is the formal account of that cause: the separating trace and the enforceability theorem classify which orders of a fixed component library desynchronize, which orders detect the disagreement, and which orders admit a faithful online boundary enforcer.

\section{Conclusion}

Layer order semantics turns a pipeline into a mathematical object: an ordered chain of trace editors observed by a decision automaton. Its organizing claim is that HTTP request smuggling is an instance of the theory: a desynchronization attack witnesses that two layer orders recognize different languages, and the worked model's separating trace coincides, under a concrete deterministic parser-pair semantics, with the boundary-disagreement condition of the documented CL.TE, TE.CL, TE.TE, and HTTP/2-downgrade families. Detection, exploitability, and impact form a refinement chain over that disagreement. The strengthened theory gives four boundaries. Regular recognition strictly contains faithful enforcement; edit automata and layer-order chains are equivalent as monolithic editors while layer order semantics classifies named decomposition and reordering; abstraction transfers concrete security coverage exactly when violating concrete traces cannot collide with accepted abstract traces; and reorder reduction is sound at the contextual level.

The result is a precise cybersecurity semantics for ordered layers. Evidence is born at one layer, transformed by later layers, preserved or erased by marker discipline, judged by a synthesized decision automaton, and classified into canonical reorder classes by minimized recognizers. The theory is finite-state throughout, gives witnesses when orders differ, and states exactly which regular marker policies can be recognized, edited, synthesized, and faithfully enforced.

\end{document}